\documentclass[conference]{IEEEtran}

\usepackage{cite}      

\usepackage[usenames,dvipsnames,svgnames,table]{xcolor}
\usepackage{graphicx}  

\usepackage{lipsum}
\usepackage{footnote}
\usepackage{psfrag}    

\usepackage{url}       
\usepackage{tikz}

\usepackage{siunitx}

\usepackage{url} 
\usepackage{amsmath,calc}
\usepackage{tikz}
\usepackage{pgfplots}
\usepackage{siunitx}
\usepackage{tikz}
\usepackage{pgfplots}
\pgfplotsset{width=7cm,compat=1.3}

\usepackage{psfrag}
\usepackage{tikz}
\usetikzlibrary{matrix}
\usepackage{pgfplots}
\usepackage{mathtools}
\usepackage{siunitx}
\usepackage[normalem]{ulem}
\usepackage{amsmath}
\newcommand{\stkout}[1]{\ifmmode\text{\sout{\ensuremath{#1}}}\else\sout{#1}\fi}

\usepackage{amsmath}   
\usepackage{amssymb}
\usepackage{algorithm}
\usepackage{algpseudocode}
\usepackage{pifont}
\usepackage{xfrac}
\newcommand{\shorten}[1]{}
\newtheorem{proposition}{Proposition}

\newtheorem{definition}{Definition}

\newtheorem{lemma}{Lemma}

\newtheorem{example}{Example}

\newcommand{\signed}%
    {{\unskip\nobreak\hfill\penalty50
      \hskip2em\hbox{}\nobreak\hfil $\blacksquare$
      \parfillskip=0pt \finalhyphendemerits=0 \par}}
\newenvironment{proof}[1]
    {
    \bf{Proof:}\rm{\noindent{#1 }}\ignorespaces
    }
    {\signed\addvspace\medskipamount}
\usepackage{color, colortbl}
\definecolor{LightCyan}{rgb}{0.88,1,1}
\definecolor{Gray}{gray}{0.9}
\definecolor{Gray1}{rgb}{0.94, 0.97, 1.0}
\definecolor{Gray2}{rgb}{0.7, 0.75, 0.71}

\begin{document}

\title{An Explicit Construction of Systematic MDS Codes with Small Sub-packetization for All-Node Repair}

%

\author{\IEEEauthorblockN{
Katina Kralevska and Danilo Gligoroski
}
\IEEEauthorblockA{
Dep. of Information Security and Communication Technologies, NTNU, Norwegian University of Science and Technology\\
Email: \{katinak, danilog\}@ntnu.no\\
}
}

\maketitle

\begin{abstract}
An explicit construction of systematic MDS codes, called HashTag+ codes, with arbitrary sub-packetization level for all-node repair is proposed. It is shown that even for small sub-packetization levels, HashTag+ codes achieve the optimal MSR point for repair of any parity node, while the repair bandwidth for a single systematic node depends on the sub-packetization level. Compared to other codes in the literature, HashTag+ codes provide from 20\% to 40\% savings in the average amount of data accessed and transferred during repair.
\end{abstract}


%
\IEEEpeerreviewmaketitle

\section{Introduction} \label{intro}

Redundancy is essential to ensure reliability in distributed storage systems. 
Maximum Distance Separable (MDS) codes are optimal erasure codes in terms of the redundancy-reliability tradeoff. In particular, a \emph{$(n, k)$ MDS code} tolerates the maximum number of failures, up to $r=n-k$ failed nodes, for the added redundancy of $r$ nodes. A \emph{systematic} $(n, k)$ MDS code is applied in such a way that the original data is equally divided into $k$ parts without encoding and stored into $k$ nodes, called systematic nodes, and $r$ linear combinations of the $k$ parts are stored into $r$ nodes, called parity nodes. In addition to their redundancy-reliability optimality, systematic MDS codes are preferred in practical systems because data access from the systematic nodes can be done instantly without decoding.

Conventional MDS codes do not perform well in terms of the \emph{repair bandwidth} defined as the amount of data that is transferred during a node repair. 
Dimakis et al. \cite{5550492} proved that the lower bound of the repair bandwidth $\gamma$ for a single node with a $(n, k)$ MDS code is:
\begin{equation}
\gamma_{MSR}^{min} \geq \frac{M}{k} \frac{n-1}{n-k},
\label{optimal}
\end{equation}
where $M$ is the file size.
The equality is met when a fraction of $\sfrac{1}{r}$-th of the stored data is transferred from all $n-1$ non-failed nodes. Minimum Storage Regenerating (MSR) codes satisfy the equality and they operate at the MSR point.

The exponential sub-packetization level is a fundamental limitation of any high-rate MSR code. The sub-packetization levels are $\alpha=r^{\sfrac{k}{r}}$ and $\alpha=r^{\sfrac{n}{r}}$ for optimal repair of systematic nodes and optimal repair of both systematic and parity nodes (all-node repair) \cite{6737213}, respectively. Large sub-packetization levels bring multiple practical challenges such as high I/O, high repair time, expensive computations, and difficult management of meta-data.
Thus constructing high-rate MDS codes with small sub-packetization levels has attracted a lot of attention in the recent years. Table \ref{comparison1} summarizes several high-rate MDS codes with small sub-packetization \cite{7949040, 7463553, 8025778, DBLP:journals/corr/abs-1709-08216,210530}. Three piggyback designs were presented in \cite{7463553}. For the purpose of this paper, we compare with piggyback design 2 that optimizes all-node repair for $r\geq 3$ and sub-packetization of $(2r-3)m$ where $m\geq 1$. HashTag codes \cite{7463553,8025778} repair the systematic nodes with the lowest repair bandwidth in the literature for arbitrary sub-packetization $2\leq \alpha\leq r^{\lceil\sfrac{k}{r}\rceil}$. Rawat et al. presented two approaches for all-node repair in \cite{DBLP:journals/corr/abs-1709-08216}. The second approach, that is more flexible in terms of the sub-packetization, requires MSR codes and error correcting codes with specific parameters to obtain $\epsilon$-MSR codes. However, codes with such specific parameters may not always be available. Additionally, there is a tradeoff between $\epsilon$ and the length of the code. Clay codes were recently presented in \cite{210530}. They are optimized for all-node repair. However, Clay codes require an exponential sub-packetization level, and for sub-packetization levels lower than the maximal exponential value, they are just MDS codes that do not achieve the optimal MSR point neither for the data nodes nor for the parity nodes.
It is observed in \cite{Rashmi:2014:HGF:2619239.2626325} that 98.08\% of the failures in Facebook's data-warehouse cluster that consists of thousands of nodes are single failures. Thus, we optimize the repair for single failures of any systematic or parity node.

\renewcommand{\arraystretch}{1.5}
\begin{table*}[t]
	\caption{Comparison of HashTag+ codes with existing MDS codes with small sub-packetization for $n-1$ helper nodes. } \label{comparison1}
		\begin{center}
		\begin{tabular}{|p{20mm}|l|p{15mm}|p{20mm}|p{23mm}|l|p{22mm}|}
	\hline
	Code & Systematic & Explicit \newline construction & Number of parities $r$  & Sub-packetization \newline $\alpha$ & All-node repair & Optimal parity \newline repair for small $\alpha$ \\
	\hline
	Piggyback 2 \cite{7949040} & Yes & Yes & $r \geq 3$ & $(2r-3)m, m\geq 1$ & Yes & No\\
	\hline
	HashTag \cite{8025778} & Yes & Yes & $r\geq 2$ & $2\leq \alpha \leq r^{\lceil \sfrac{k}{r}\rceil}$ & No & No\\ 
	\hline
		Rawat et al. \cite{DBLP:journals/corr/abs-1709-08216} & Yes & Yes & \textbf{$r \geq 2$} & $r^\tau, \tau\geq 1$ & Yes & No\\
	\hline
	Clay codes \cite{210530} & Yes & Yes & $r\geq 2$ & $\alpha \leq r^{\sfrac{n}{r}}$ & Yes & No\\
	\hline
	\textbf{HashTag+} & Yes & Yes & $r\geq 2$ & $4\leq \alpha \leq r^{\lceil\sfrac{n}{r}\rceil}$ & Yes & Yes\\
	\hline
		\end{tabular}
\end{center}
\end{table*}

In this paper we present a family of MDS codes called HashTag+ codes with the following properties: 1. They are systematic MDS codes; 2. They are exact-repairable codes; 3. They have a high-rate; 4. They have a flexible sub-packetization ($4 \leq \alpha\leq r^{\lceil\sfrac{n}{r}\rceil}$); 5. They achieve the MSR point for repair of single parity node for sub-packetization levels lower than or equal to the maximal exponential value of $r^{\lceil\sfrac{n}{r}\rceil}$; 6. They achieve the MSR point for repair of single systematic node for $\alpha=r^{\lceil\sfrac{n}{r}\rceil}$ and repair near-optimally for $\alpha<r^{\lceil\sfrac{n}{r}\rceil}$; 7. They are access-optimal (access and transfer the same amount of data). We combine the framework proposed by Li et al. \cite{8006804} and the family of MDS codes called HashTag codes \cite{8025778}. Compared to the work by Li et al. \cite{8006804} where they focus on MSR codes with the maximal sub-packetization level $\alpha = r^{\lceil\frac{n}{r}\rceil}$, we construct explicit codes for the whole range of sub-packetization levels $4\leq \alpha \leq r^{\lceil\frac{n}{r}\rceil}$ motivated by the practical importance of codes with small sub-packetization levels.


The rest of the paper is organized as follows. Section \ref{code} presents HashTag+ code construction by first giving two examples and then presenting a general algorithm and performance comparison between HashTag+ and state-of-the-art codes. Section \ref{conc} concludes the paper.

\textbf{\textit{Notations}}. For two integers $0<i<j$, we denote the set $\{i,i+1,\ldots, j \}$ by $[i:j]$, while the set $\{0,1,\ldots, j-1 \}$ is denoted by $[j]$. 
Vectors and matrices are denoted with a bold font.

\section{HashTag+ Code Construction} \label{code}
We now present two examples of HashTag+ codes with the maximal and a small sub-packetization, and we later give algorithms for general code construction and repair. An appealing feature of HashTag+ codes is that they support any values of code parameters $k$, $r\geq 2$, and sub-packetization $4\leq \alpha \leq r^{\lceil\frac{n}{r}\rceil}$ including cases where $r$ does not divide $n$. 

\begin{example}
Consider a $(6, 4)$ HashTag MDS code with $\alpha=2^{\sfrac{4}{2}}=4$ as a base code. The code given in Fig.\ref{general0} is generated with Alg. 1 from \cite{8025778} where the coefficients are from the finite field $\mathbb{F}_{16}$ with irreducible polynomial $x^4+x^3+1$. This code achieves the bound in Eq. (\ref{optimal}) for repair of any single systematic node, i.e., 10 symbols are read and transferred for repair of 4 symbols of any systematic node. The repair of a single parity node is the same as Reed-Solomon codes, i.e., 16 symbols are read and transferred for repair of 4 parity symbols. The goal is to construct a code that provides optimal repair of the parity nodes as well (all-node repair).

\begin{figure*}
	    \centering
	\includegraphics[width=7in]{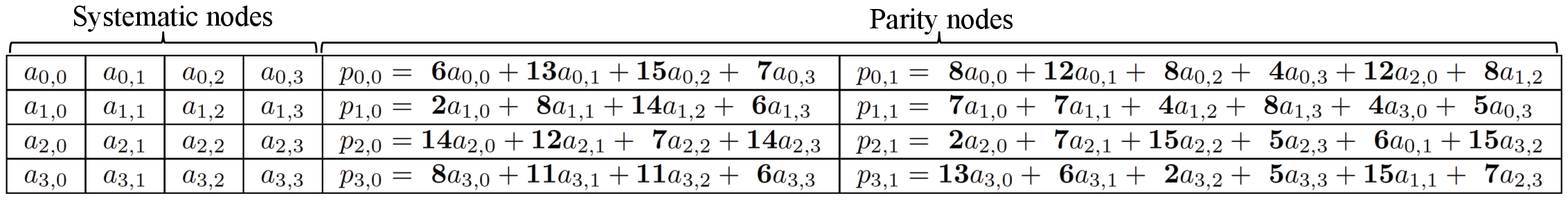}
	    \caption{A systematic $(6, 4)$ HashTag MDS code with $\alpha=4$.}
	    \label{general0}
	\end{figure*}	

In a first step, we generate $r-1=1$ additional instances of the $(6, 4)$ HashTag MDS code with $\alpha=4$ by using Alg. 1 from \cite{8025778}. The data of the first systematic node $d_0$ stored in instance $0$ is $a_{0,0}, \ldots, a_{3,0}$ and in instance $2$ is $a_{4,0}, \ldots, a_{7,0}$ as it is shown in Fig. \ref{general1}. In this way, we obtain a $(6, 4)$ code with sub-packetization level of $r\times \alpha = 2\times 4 = 8$. A systematic node $d_j$, $j=0,\ldots,3$, comprises the symbols $a_{i,j}$ from the two instances where $i=0,\ldots, 7$ and $j=0,\ldots,3,$ and a parity node $p_l$, $l=0,1,$ comprises the symbols $p_{i,l}$ from the two instances where $i=0,\ldots, 7$ and $l=0,1$. Note that the base code from above has been renamed to instance 0.
	\begin{figure*}[h!]
	    \centering
	\includegraphics[width=7in]{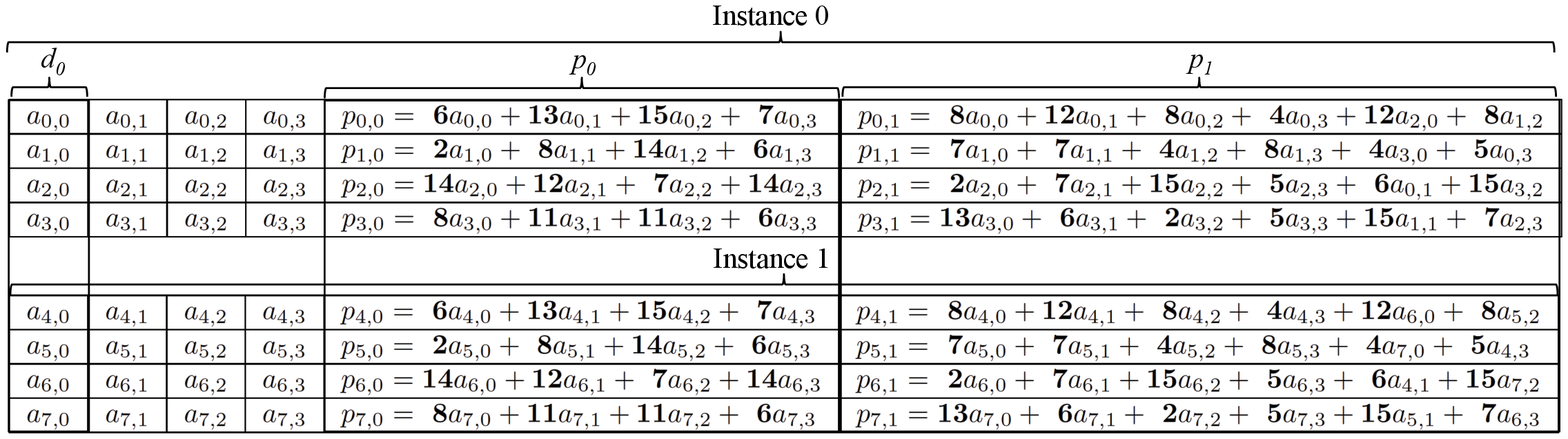}
	\caption{Two instances of a $(6, 4)$ HashTag code with $\alpha = 8$.}
	\label{general1}
	\end{figure*}

In the second step, we permute the data in the two instances of the parity nodes $p_0$ and $p_1$. First, this data is represented as $p_l^{(i)}$ where the superscript $i$ denotes the instance and the subscript $l$ denotes the parity node. 
Then, the permutation is as follows: $p_l^{(i)}\rightarrow p_{l+i}^{(i)}$ where the index arithmetic is cyclic, i.e., modulo $r$ (for example $l+i=3\rightarrow l+i=1$).
    \vspace{0.2cm}
	\begin{center}
	\includegraphics[width=2.3in]{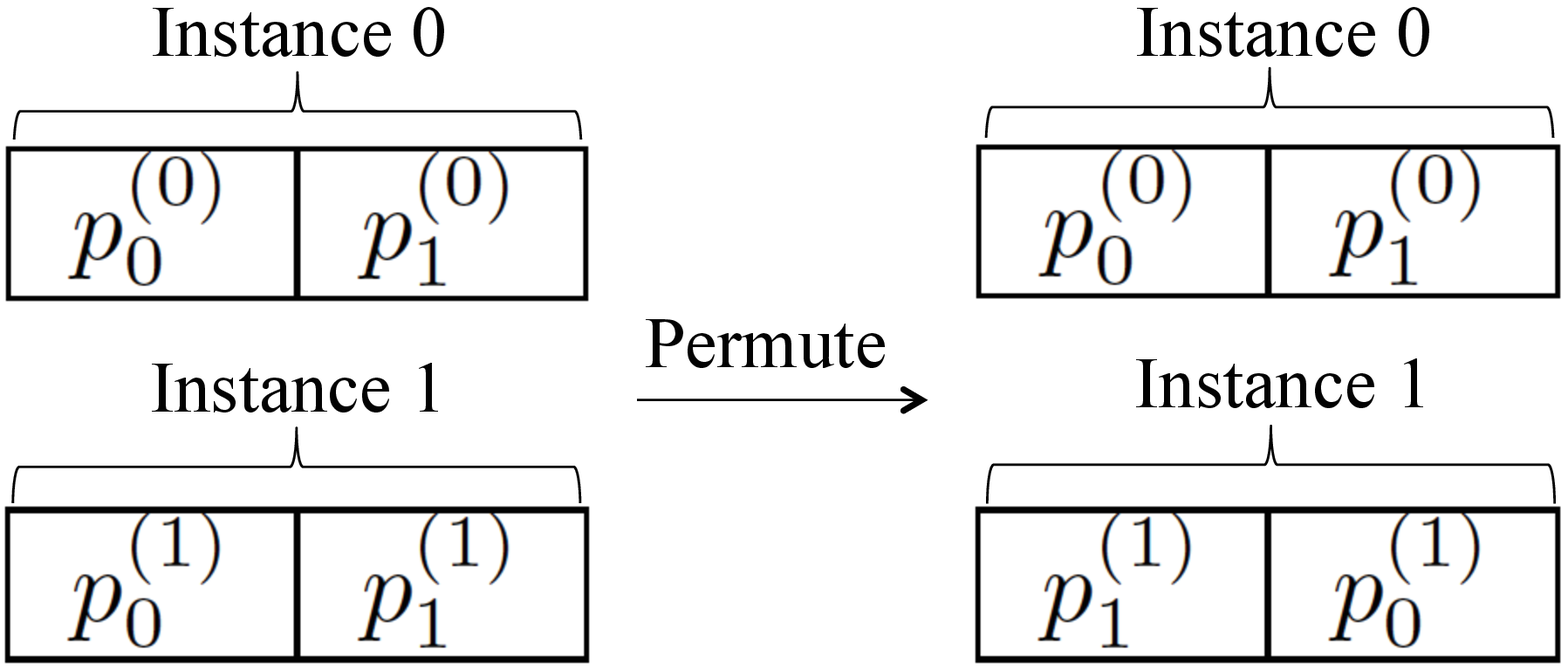}
	\label{general2}
	  \vspace{0.2cm}
\end{center}

In the third step, the data from the parity nodes is paired following this rule:
\begin{equation}\label{eq:PairingRule}
p_l^{(i)} = \left\{
\begin{array}{ll}
p_l^{(i)}, & \text{if } i = l,\\
\theta_{l, i} p_l^{(i)} + p_i^{(l)}, & \text{ otherwise }
\end{array} \right.
\end{equation} where $\{\theta_{l, i}, \theta_{i,l}\} \subseteq \{1, \theta \}$ and $\theta \in \mathbb{F}_{16} \setminus \{0, 1\}$. 
The bidirectional arrows in the figure below shows which parity parts are paired together. This completes the code generation. 
\vspace{0.2cm}
\begin{center}
	\includegraphics[width=2.6in]{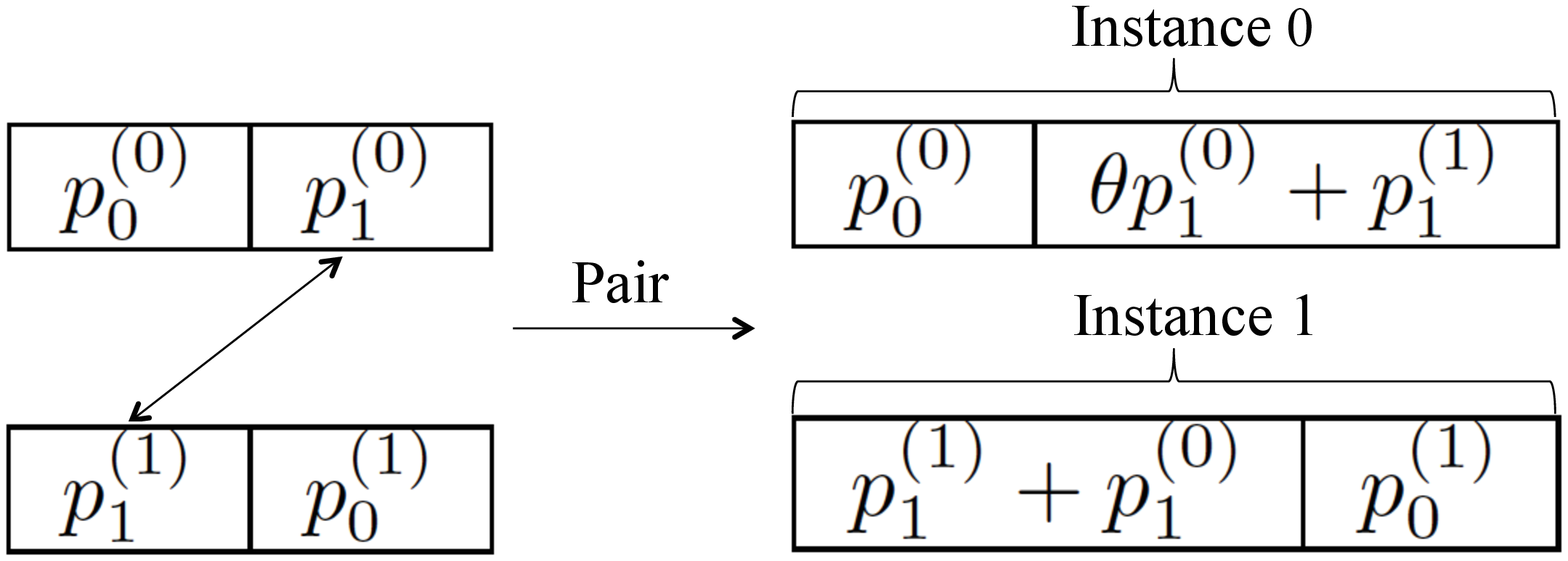}
	\label{general3}
	  \vspace{0.2cm}
\end{center}

The final $(6, 4)$ HashTag+ code with $\alpha=8$ that provides optimal all-node repair is given in Fig. \ref{final8stripes}. 

We now illustrate that this code recovers optimally any systematic or parity node. Let us assume that node $d_0$ has failed. In order to recover $a_{0,0}, a_{1,0}$, we transfer 6 symbols $a_{0,1}, a_{1,1}, a_{0,2}, a_{1,2}, a_{0,3}, a_{1,3}$ from instance 0 of the non-failed systematic nodes and 2 non-paired symbols $p_{0,0}, p_{1,0}$ from the parity nodes. Next we recover $a_{4,0}, a_{5,0}$ by downloading 6 symbols $a_{4,1}, a_{5,1}, a_{4,2}, a_{5,2}, a_{4,3}, a_{5,3}$ from instance 1 of the systematic nodes and 2 non-paired symbols $p_{4,0}, p_{5,0}$ from the parity nodes.
To recover the remaining symbols $a_{2,0}, a_{3,0}$ from instance 0, we transfer the paired symbols $\theta p_{0,1}+p_{4,1}, p_{4,1}+p_{0,1}$ and solve $2\times 2$ system of linear equations. In a similar manner we recover the last two symbols $a_{6,0}, a_{7,0}$ by transferring the paired symbols $p_{5,1}+p_{1,1}, \theta p_{1,1}+p_{5,1}$. Thus, the repair of $d_0$ (or any other systematic node)  requires 20 symbols in total, and it achieves the bound in Eq.(\ref{optimal}).

The same amount of data is transferred when repairing the parity nodes $p_0$ or $p_1$. We first repair the unpaired symbols $p_{0,0}, p_{1,0}, p_{2,0}, p_{3,0}$ from instance 0 by transferring all 16 symbols from instance 0 of the systematic nodes $a_{0,0}, a_{1,0}, a_{2,0}, a_{3,0}, a_{0,1}, a_{1,1}, a_{2,1}, a_{3,1}, a_{0,2}, a_{1,2}, a_{2,2}, a_{3,2},$\newline $a_{0,3}, a_{1,3}, a_{2,3}, a_{3,3}$. Next the paired symbols from $p_0$ are recovered by downloading the 4 symbols from instance 0 of $p_1$. In total, 20 symbols are read and transferred for repair of 8 symbols from $p_0$. 

\begin{figure*}[t!] \label{final8stripes} 
	\begin{center}
	    \vspace{0.1cm}
		\includegraphics[width=7.1in]{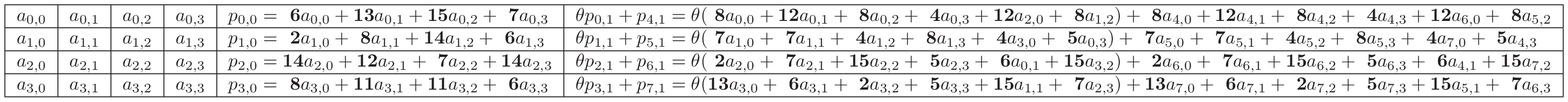}
	\end{center}
	\begin{center}
	    \vspace{0.1cm}
		\includegraphics[width=7.1in]{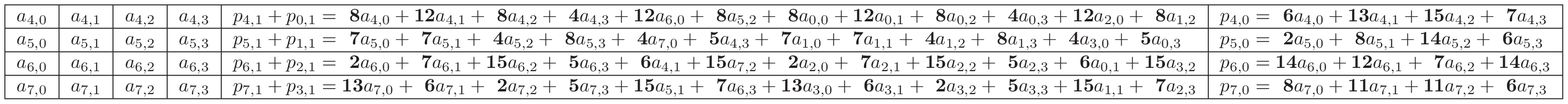}
		\caption{Two instances of a $(6, 4)$ HashTag+ code with $\alpha=8$ where $\theta \in \mathbb{F}_{16} \setminus \{0, 1\}$. }
	\end{center}
\end{figure*}
\end{example}

\begin{example}
We next give a $(6, 4)$ HashTag+ code with $\alpha=4$ in Fig. 4. The code is obtained by following the steps from the previous example where the base code is a $(6, 4)$ HashTag code with $\alpha=2$. Note that the sub-packetization level in this example is lower than the optimal one in Example 1. The goal is to illustrate that the code achieves the MSR point when repairing a single parity node although the sub-packetization is small. 

Repairing any systematic node is near-optimal, i.e., 12 symbols for repair of 4 symbols. Let us assume that node $d_0$ has failed. In order to recover $a_{0,0}$, we transfer 3 symbols $a_{0,1}, a_{0,2}, a_{0,3}$ from instance 0 of the non-failed systematic nodes and 1 non-paired symbol $p_{0,0}$ from the parity node $p_0$. Next we recover $a_{2,0}$ by downloading 3 symbols $a_{2,1}, a_{2,2}, a_{2,3}$ from instance 1 of the systematic nodes and 1 non-paired symbol $p_{2,0}$ from the parity node $p_1$.
To recover the remaining symbols $a_{1,0}, a_{3,0}$ from instance 0 and 1, we transfer the paired symbols $\theta p_{0,1}+p_{2,1}, p_{2,1}+p_{0,1}$ and $a_{1,2}, a_{3,2}$ (due to the small sub-packetization level) and solve $2\times 2$ system of linear equations. Thus, the repair of $d_0$ (or any other systematic node) requires 12 symbols in total, and the repair bandwidth of the systematic nodes is the same as that of the base code (HashTag code).

However, the repair bandwidth for any parity node achieves the lower bound in Eq.(\ref{optimal}). In particular, all 4 symbols from $p_0$ are repaired by transferring all 8 symbols from instance 0 of the systematic nodes $a_{0,0}, a_{1,0}, a_{0,1}, a_{1,1}, a_{0,2}, a_{1,2}, a_{0,3}, a_{1,3}$ and 2 symbols $\theta p_{0,1}+p_{2,1}$ and $\theta p_{1,1}+p_{3,1}$ from instance 0 of $p_1$. In total, 10 symbols are read and transferred for repair of 4 symbols from $p_0$. Repair of $p_1$ requires the same amount of repair bandwidth.

\begin{figure*} \label{final4substripes} 
	\begin{center}
	    \vspace{0.1cm}
		\includegraphics[width=7.1in]{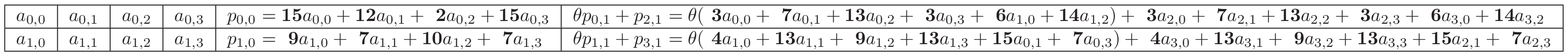}
	\end{center}
	\begin{center}
	    \vspace{-0.3cm}
		\includegraphics[width=7.15in]{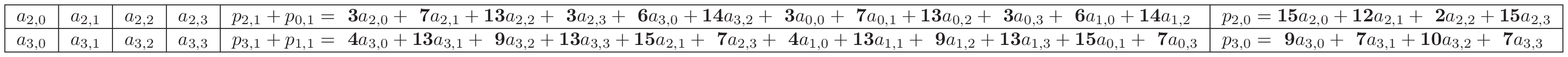}
		\caption{Two instances of a $(6, 4)$ HashTag+ code with $\alpha=4$ where $\theta \in \mathbb{F}_{16} \setminus \{0, 1\} $.}
	\end{center}
\end{figure*}
\end{example}

\subsection{General Code Construction}
Consider a file of size $M = k \alpha$ symbols from a finite field $\mathbb{F}_q$ stored in $k$ systematic nodes $d_j$ of capacity $\alpha$ symbols. We start the construction with a HashTag code \cite{7463553,8025778} as a base code that is defined as follows.

\begin{definition}\label{HashTagCodes}
	A $(n,k)_q$ HashTag linear code is a vector systematic code defined over an alphabet $\mathbb{F}_q^\alpha$ for some $2 \leq \alpha \leq r^{\lceil\sfrac{k}{r}\rceil}$. It encodes a vector $\mathbf{x} = (\mathbf{x}_0,\ldots,\mathbf{x}_{k-1})$, where $\mathbf{x}_i = (x_{0,i}, x_{1,i},\ldots,x_{\alpha-1,i})^T \in \mathbb{F}_q^\alpha$ for $i \in [k] $, to a codeword $\mathcal{C}(\mathbf{x}) = \mathbf{c} = (\mathbf{c}_0, \mathbf{c}_1, \ldots, \mathbf{c}_{n-1})$ where the systematic parts $\mathbf{c}_i=\mathbf{x}_i$ for $i \in [k]$ and the parity parts $\mathbf{c}_i=(c_{0,i}, c_{1,i},\ldots,c_{\alpha-1,i})^T$ for $i \in [k:n-1]$ are computed by the linear expressions that have a general form as follows:
	\begin{equation}\label{LinEquations}
	c_{j,i}=\sum f_{\nu,j, i} x_{j_1,j_2},\\
	\end{equation}
	where $f_{\nu,j, i}\in \mathbb{F}_q$ and the index pair $(j_1,j_2)$ is defined in the $j$-th row of the index array $\mathbf{P}_{i-r-1}$ where $\nu \in [r]$. The $r$ index arrays $\mathbf{P}_0,\ldots,\mathbf{P}_{r-1}$ are defined as follows:
	\begin{equation*}
	\hspace{-0.7cm}
	\mathbf{P}_0=
	\begin{bmatrix}\footnotesize
	(0, 0) & (0, 1) & \ldots & (0, k-1)\\
	(1, 0) & (1, 1) & \ldots & (1, k-1)\\
	\vdotswithin{1} & \vdotswithin{\alpha_n} & \ddots & \vdotswithin{{\alpha_n}^{k-1}}\\
	(\alpha-1, 0) & (\alpha-1, 1) & \ldots & (\alpha-1, k-1)\\
	\end{bmatrix},
	\end{equation*}
	$$\ \ \ \ \ \ \ \ \ \ \ \ \ \ \ \ \ \ \ \ \ \ \ \ \ \ \ \ \ \ \ \ \ \ \ \ \ \ \ \ \ \overbrace{\ \ \ \ \ \ \ \ \ \ \ \ \ \ \ \ \ \ \ \ }^{\lceil \frac{k}{r} \rceil}$$
	\begin{equation*}\footnotesize
	\mathbf{P}_i=
	\begin{bmatrix}
	(0, 0) & \ldots & (0, k-1) &  (?, ?) & \ldots & (?, ?) \\
	(1, 0) & \ldots & (1, k-1) & (?, ?) & \ldots & (?, ?) \\
	\vdotswithin{1} & \vdotswithin{\alpha_n} & \ddots & \vdotswithin{{\alpha_n}^{k-1}}\\
	(\alpha-1, 0) & \ldots & (\alpha-1, k-1) & (?, ?) & \ldots & (?, ?) \\
	\end{bmatrix}.
	\end{equation*}
	where the values of the indexes $(?, ?)$ are determined by a scheduling algorithm that guarantees the code is MDS, i.e. the entire information $\mathbf{x}$ can be recovered from any $k$ out of the $n$ vectors $\mathbf{c}_i$. In addition, the algorithm ensures optimal or near-optimal repair by scheduling the indexes of the elements from $\mathbf{x}_i$ into $\lceil\sfrac{\alpha}{r}\rceil$ rows in the $r-1$ index arrays $\mathbf{P}_j$ where $j=1,\ldots,r-1$. $\blacksquare$
\end{definition}
The scheduling algorithm for Def. \ref{HashTagCodes} is presented in \cite{7463553,8025778}. Note that in the original presentation the indexing of the arrays is from $1$ to $r$ but in order to synchronize with the transformation of Li et al. \cite{8006804} here we use the indexing of the arrays from 0 to $r-1$.
The set of all symbols in $d_j$ is partitioned in disjunctive subsets where at least one subset has $\lceil \sfrac{\alpha}{r}\rceil$ number of elements.
The set of indexes $D=\{1,\ldots,\alpha\}$, where the $i-$th index of $a_{i,j}$ from $d_j$ is represented by $i$ in $D$, is partitioned in $r$ disjunctive subsets $D = \cup_{\rho=1}^{r}D_{\rho,d_j}$ where at least one subset has $\lceil \sfrac{\alpha}{r}\rceil$ elements. One subset $D_{\rho,d_j}$ is assigned per disk. The indexes in $D_{\rho,d_j}$ are the row positions where the pairs $(i,j)$ with indexes $i \in \mathcal{D} \setminus D_{\rho,d_j}$ are scheduled (the zero pairs are replaces with concrete $(i, j)$ pairs).
By using the code defined in Def. 1 as a base code, we next define HashTag+ code.
\begin{definition}
A $(n,k)_q$ HashTag+ linear code is a vector systematic code defined over an alphabet $\mathbb{F}_q^\alpha$ for some $4 \leq \alpha \leq r^{\lceil\sfrac{n}{r}\rceil}$.	
\end{definition}

The algorithm for constructing a $(n, k)$ HashTag+ code is given in Alg. 1.
\begin{algorithm}
	\small
	\caption{HashTag+ code construction
		\newline
		\textbf{Input:} $(n, k)$ HashTag code with sub-packetization $\alpha$
		\newline
		\textbf{Output:} $(n, k)$ HashTag+ code with sub-packetization $r \times \alpha$}
	\label{AlgConstruct}
	\begin{algorithmic}[1]
		\State Construct $r-1$ additional instances of a $(n, k)$ HashTag code with sub-packetization $\alpha$;
		\State Permute the data from the $i$-th instance in the $l$-th parity node as $p_l^{(i)}\rightarrow p_{l+i}^{(i)}$;
		\State Compute the parity parts $p_{l}^{(i)}$ with the rule in Eq.(\ref{eq:PairingRule}).
	\end{algorithmic}
\end{algorithm}

 The construction of HashTag+ codes given in Alg. 1 is sound and there always exists a finite field $\mathbb{F}_q$ and a set of non-zero coefficients from the field such that the HashTag+ code is MDS due to the following Lemma:
\begin{lemma}\label{MDS}
	There exists a choice of non-zero coefficients $c_{l,i,j}$ where $l=1, \ldots, r,$ $i=1, \ldots, \alpha$ and $j=1, \ldots, k$ from $\mathbb{F}_q$ such that the code is MDS if $q \geq \binom {n} {k} r \alpha$.
\end{lemma}

\begin{proof}
	It is sufficient to combine Theorem 1 from \cite{8025778} about the base HashTag codes and Theorem 2 and 3 from \cite{8006804}. Namely, Theorem 1 from \cite{8025778} guarantees that the size of the finite field for the base HashTag code is sufficient to be $q \geq \binom {n} {k} r \alpha$ in order to find a HashTag MDS code. Then, according to Theorem 2 and 3 from \cite{8006804} the HashTag+ code has optimal repair bandwidth, has optimal rebuilding access and is a MDS code.
\end{proof}


\subsection{Repair of systematic nodes}
Alg. \ref{AlgRepair} shows the repair of a systematic node where the systematic and the parity nodes are global variables.
A set of $\lceil \sfrac{\alpha}{r^2} \rceil$ symbols is accessed and transferred from all $n-1$ non-failed nodes from each instance. 

\begin{proposition} \label{bw}
	The repair bandwidth for a single systematic node $\gamma_s$ is bounded between the following lower and upper bounds:
	\begin{equation}
	\label{LowerUpperBounds}
	\frac{(n-1)}{\alpha} \lceil \frac{\alpha}{r} \rceil \leq \gamma_s \leq \frac{(n-1)}{\alpha} \lceil \frac{\alpha}{r} \rceil + \frac{(r-1)}{\alpha} \lceil \frac{\alpha}{r} \rceil
	\lceil \frac{k}{r} \rceil.
	\end{equation}
\end{proposition}

\begin{proof}
	We read in total $k \lceil \frac{\alpha}{r}\rceil$ elements in the first for loop of Alg. \ref{AlgRepair}. Additionally, $(r-1)\lceil \frac{\alpha}{r}\rceil$ elements are read in Step 7 of the second for loop. Assuming that we do not read more elements in Step 6, we determine the lower bound as $k \lceil \frac{\alpha}{r}\rceil + (r-1)\lceil \frac{\alpha}{r}\rceil = (n-1)\lceil \frac{\alpha}{r}\rceil$ elements, i.e., the lower bound is $\frac{(n-1)}{\alpha} \lceil \frac{\alpha}{r}\rceil$ (since every element has a size of $\frac{1}{\alpha}$).
	To derive the upper bound, we assume that we read all elements $a_{i,j}$ from the extra $\lceil \frac{k}{r}\rceil $ columns of the arrays $\mathbf{P}_0, \ldots, \mathbf{P}_{r-1}$ in Step 6. Thus, the upper bound is $\frac{(n-1)}{\alpha} \lceil \frac{\alpha}{r} \rceil + \frac{(r-1)}{\alpha} \lceil \frac{\alpha}{r} \rceil
	\lceil \frac{k}{r} \rceil$.
\end{proof}
\begin{algorithm}
	\small
	\caption{Repair of systematic node $d_j$ 
		\newline
		\textbf{Input:} $j$ (where $j=0,\ldots, k-1$);
		\newline
		\textbf{Output:} $d_j$;
	    \newline
		\textbf{Note:} All indexes $i$ are determined by the expression $i \in D_{\rho,d_j}$}
	\label{AlgRepair}
	\begin{algorithmic}[1]
		\For {$v=0, v < r$}
		\State Access and transfer $(k-1) \lceil \sfrac{\alpha}{r^2}\rceil$ symbols $a_{i,j}$ from the $v$-th instance of all $k-1$ non-failed systematic nodes and $\lceil \sfrac{\alpha}{r^2}\rceil$ non-paired symbols $p_{i,j}$ from the $v$-th instance of the parity nodes;
		\State Repair $a_{i,j}$ from the $v$-th instance;
		\EndFor
		\For {$v=0, v < r$}
		\State Access and transfer the symbols $a_{i,l}$ from the $v$-th instance listed in the $i-$th row of the arrays $\mathbf{P}_0, \ldots, \mathbf{P}_{r-1}$ that have not been read in Step 2;
		\State Access and transfer $(r-1)\lceil \sfrac{\alpha}{r^2}\rceil$ paired symbols $p_{i,j}$ from the $v$-th instance;
		\State Repair $a_{i,j}$ by solving paired $r\times r$ linear systems of equations.
		\EndFor
	\end{algorithmic}
\end{algorithm}

\subsection{Repair of parity nodes}

Repair of a single parity node is given in Alg. \ref{AlgRepairP}.
\begin{algorithm}
	\small
	\caption{Repair of a parity node $p_l$ where $l=0,\ldots, r-1$
		\newline
		\textbf{Input:} $l$;
		\newline
		\textbf{Output:} $p_l$.}
	\label{AlgRepairP}
	\begin{algorithmic}[1]
		\State Access and transfer all symbols from instance $l$ of the systematic nodes and the non-failed parity nodes;
		\State Repair the symbols from $p_l$.
	\end{algorithmic}
\end{algorithm}

Without a proof (just a reference to Theorem 2 and 3 from \cite{8006804}) we give the following Proposition:
\begin{proposition}
	The repair bandwidth for a single parity node $\gamma_p$ reaches the lower bound given in Eq. (\ref{optimal}) for any sub-packetization level $\alpha$ including small $\alpha$, i.e.,
	\begin{equation}
	\label{LowerBound}
	\gamma_p = \frac{(n-1)}{\alpha} \lceil \frac{\alpha}{r} \rceil.
	\end{equation}
\end{proposition}

\subsection{Performance Analysis} \label{perf}
We compare the average amount of data read and downloaded during a repair of a single node taking into account all nodes (systematic and parity nodes). HashTag+ codes outperform both Piggyback 2 and HashTag codes for any code parameters as it is shown in Fig. \ref{piggy}. Compared to HashTag codes, the lower repair bandwidth comes at the cost of an increased sub-packetization of factor $r$. HashTag+ codes offer savings of up to 40\% in the average amount of data accessed and transferred during repair compared to Piggyback 2.


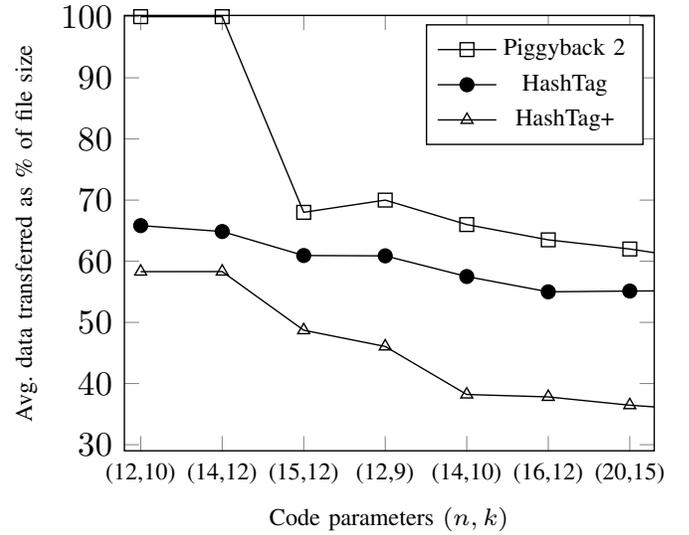
\begin{figure}
	\centering
	\begin{tikzpicture}[scale=1.3]
	\begin{axis}[
	xlabel=\text{\scriptsize Code parameters $(n, k)$},
	ylabel={\scriptsize Avg. data transferred as \% of file size},
	xmin=-0.2, xmax=6.3,
	ymin=29, ymax=100.2,
	xtick={0,1,2,3,4,5,6},
	ytick={30, 40, 50, 60, 70, 80, 90, 100},
	xticklabels={\scriptsize (12$\text{,}$10), \scriptsize (14$\text{,}$12), \scriptsize (15$\text{,}$12), \scriptsize (12$\text{,}$9), \scriptsize (14$\text{,}$10), \scriptsize (16$\text{,}$12), \scriptsize (20$\text{,}$15), \scriptsize (24$\text{,}$18)}
	]
	\addplot[
	color=black,
	mark=square,
	]
	coordinates {
		(0, 100)(1, 100)(2, 68)(3, 70)(4, 66)(5, 63.5)(6, 62)(7, 60)
	};
	
	\addplot[
	color=black,
	mark=*,
	]
	coordinates {
		(0, 65.83)(1, 64.86)(2, 60.936)(3, 60.88)(4, 57.5)(5, 55)(6, 55.125)(7, 55.208335)
	};
	\addplot[
	color=black,
	mark=triangle,
	]
	coordinates {
		(0, 58.3)(1, 58.31)(2, 48.72)(3, 46.065)(4, 38.21)(5, 37.8125)(6, 36.4575)(7, 35.53)
	};
	\legend{\scriptsize Piggyback 2\\ \scriptsize HashTag \\ \scriptsize HashTag+\\}
	
	\end{axis}
	\end{tikzpicture}
	\vspace{-0.4cm}
	\caption{Average data read and transferred for repair of any single node with Piggyback 2 \cite{7949040} for $\alpha=4\times(2r-3)$, HashTag \cite{8025778} for $\alpha=8$, and HashTag+ for $\alpha=8 \times r$.}
	\label{piggy}
\end{figure}

\section{Conclusions}\label{conc}
We presented a general construction of a family of systematic MDS codes called HashTag+ codes that reaches the lower bound of the repair bandwidth for any single failure of all nodes when $\alpha=r^{\lceil \sfrac{n}{r} \rceil}$. HashTag+ codes have a high-rate and they have a flexible sub-packetization level ($4 \leq \alpha\leq r^{\lceil\sfrac{n}{r}\rceil}$). They also achieve the MSR point for repair of single parity node for sub-packetization levels lower than or equal to the maximal exponential value of $r^{\lceil\sfrac{n}{r}\rceil}$. Additionally they are access-optimal i.e. they access and transfer the same amount of data.

HashTag+ codes are the first explicit construction in the literature that repairs optimally the parity nodes even for small sub-packetization levels.
The repair bandwidth for the systematic nodes is as close as possible to the lower bound when $\alpha < r^{\lceil \sfrac{n}{r} \rceil}$.  

\bibliographystyle{IEEEtran}
\bibliography{refer}

\end{document}